\documentclass[format=sigconf,screen=true,9pt]{acmart}

\usepackage[T1]{fontenc}
\usepackage{etoolbox,graphicx,setspace,listings,multicol,xspace,enumitem,booktabs,xcolor}
\usepackage{balance}
\usepackage{subcaption}
\usepackage{algorithm,mathtools}
\usepackage[noend]{algpseudocode}
\usepackage[skip=10pt,labelfont=bf]{caption}
\usepackage{tabularx}
\usepackage{amsmath}
\usepackage{dsfont}

\newcommand*\LSTfont{\Small\ttfamily\SetTracking{encoding=*}{-60}\lsstyle}
\newcommand*{\sn}{\textsc{TASTI}\xspace}
\newcommand*{\indname}{\textsc{TASTI}\xspace}
\newcommand*{\noscope}{\textsc{NoScope}\xspace}
\newcommand*{\blazeit}{\textsc{BlazeIt}\xspace}
\newcommand*{\supg}{\textsc{SUPG}\xspace}

\newcommand{\minihead}[1]{{\vspace{.45em}\noindent\textbf{#1.} }}
\newcommand{\miniheadit}[1]{{\vspace{.45em}\noindent\textit{#1.} }}

\newcommand*{\universe}{\mathcal{D}}
\newcommand*{\emb}{\phi}
\newcommand*{\lossq}{\ell_Q}

\newcommand*{\Exp}{\mathbb{E}}

\DeclareMathOperator*{\argmin}{arg\,min}

\newtheorem{thm}{Theorem}
\newtheorem{lemma}{Lemma}

\newtoggle{rcolors}
\togglefalse{rcolors}
\newcommand{\colora}[1]{\iftoggle{rcolors}{{\color{blue}{#1}}}{#1}}

\begin{document}

\title{Semantic Indexes for Machine Learning-based Queries over Unstructured Data}

\author{Daniel Kang*}
\affiliation{
  \institution{Stanford University}
  \city{Stanford}
  \country{USA}
}
\email{ddkang@stanford.edu}
\author{John Guibas*}
\affiliation{
  \institution{Stanford University}
  \city{Stanford}
  \country{USA}
}
\email{jtguibas@stanford.edu}
\author{Peter Bailis}
\affiliation{
  \institution{Stanford University}
  \city{Stanford}
  \country{USA}
}
\email{pbailis@stanford.edu}
\author{Tatsunori Hashimoto}
\affiliation{
  \institution{Stanford University}
  \city{Stanford}
  \country{USA}
}
\email{thashim@stanford.edu}
\author{Matei Zaharia}
\affiliation{
  \institution{Stanford University}
  \city{Stanford}
  \country{USA}
}
\email{matei@cs.stanford.edu}
\titlenote{Marked authors contributed equally.}

\begin{abstract}

Unstructured data (e.g., video or text) is now commonly queried by using
computationally expensive deep neural networks or human labelers to produce
structured information, e.g., object types and positions in video. To accelerate
queries, many recent systems (e.g., \blazeit, \noscope, Tahoma, \supg, etc.)
train a \emph{query-specific proxy model} to approximate a large \emph{target
labelers} (i.e., these expensive neural networks or human labelers). These
models return \emph{proxy scores} that are then used in query processing
algorithms.  Unfortunately, proxy models usually have to be trained \emph{per
query} and require large amounts of annotations from the target labelers.

In this work, we develop an index (trainable semantic index, \sn)
that simultaneously removes the need for per-query proxies and is more efficient
to construct than prior indexes. \sn accomplishes this by leveraging semantic similarity across
records in a given dataset. Specifically, it produces embeddings for each
record such that records with close embeddings have similar target labeler outputs.
\sn then generates high-quality proxy scores via embeddings without needing
to train a per-query proxy. These scores can be used in existing proxy-based
query processing
algorithms (e.g., for aggregation, selection, etc.). We theoretically
analyze \sn and show that a low embedding training error guarantees downstream
query accuracy for a natural class of queries. We evaluate \sn on five video,
text, and speech datasets, and three query types. We show that \indname's
indexes can be 10$\times$ less expensive to construct than generating
annotations for current proxy-based methods, and accelerate queries by up to
24$\times$.

\end{abstract}

\maketitle

\section{Introduction}

Unstructured data, such as video and text, is becoming increasingly feasible to
analyze due to deep neural networks (DNNs). A common approach is to use DNNs to
extract structured information from these sources and use the structured data to
answer queries. However, accurate DNNs can be prohibitively expensive to execute
on large data volumes. For example, Mask R-CNN \cite{he2017mask} (an object
detection DNN) can be used to extract object types and positions from frames of
video, which can be used to answer queries such as counting the number of cars
visible in a frame or finding frames with both a car and a bicycle.
Unfortunately, Mask R-CNN executes as slow as 3 frames per second (fps), or
10$\times$ slower than real time.

To reduce these costs, recent work has proposed
\emph{query-specific proxy models} to approximate high-quality \emph{target
labelers}. Low-cost proxy models can be used for selecting data records
that match a predicate, aggregation queries, and limit queries
\cite{kang2019blazeit, kang2017noscope, kang2020approximate,
anderson2018predicate, hsieh2018focus, lu2018accelerating, bastani2020vaas}. For
each query, a proxy model is trained to generate \emph{proxy scores} for data
records, in which the goal is to approximate the result of executing the target
labeler on that data record for the particular query. These scores are then used in
various algorithms depending on query type. For example, \blazeit \cite{kang2019blazeit} will train a
proxy model to (approximately) count the number of cars per frame of a video to
answer the car counting query, and selection algorithms (e.g., \noscope
\cite{kang2017noscope}, probabilistic predicates \cite{lu2018accelerating},
\supg \cite{kang2020approximate}, and Tahoma \cite{anderson2018predicate}) will
train a separate proxy model to (approximately) filter frames with cars and
bicycles for selecting such frames.

Unfortunately, methods based on query-specific proxy models have three key drawbacks.
First, obtaining
large amounts of training data from the target labeler to train the proxy models can be expensive. For
example, \blazeit and \noscope require hours of GPU compute to execute the
target labeler to produce labels to train the proxy models \cite{kang2017noscope,
kang2019blazeit} and other systems require expensive human annotations
\cite{anderson2018predicate, hsieh2018focus, lu2018accelerating}. Second, these
systems require new training procedures for each query type, which can be
difficult to develop. Third, query-specific proxy models cannot easily share
computation across different queries or query types. Thus, using proxy models
can be challenging and computationally expensive.

We observe that this prior work ignores a key opportunity: redundancy present in
the \emph{target labeler outputs} of many datasets. For example, two frames with visually distinct cars
in the bottom left would have the same result for many queries, e.g., counting
the number of cars or selecting cars in the bottom left. Namely, the structured
outputs (i.e., target labeler results) of many data records are semantically
similar. While fully computing target labeler outputs for all records is expensive,
query processing systems would ideally use this similarity to avoid repeated
work and target labeler invocations.

To address these issues and leverage this opportunity, we propose
\colora{\textbf{T}r\textbf{A}inable \textbf{S}eman\textbf{T}ic \textbf{I}ndexes} (\sn).
\sn is an indexing method over unstructured data for accelerating downstream
proxy score-based query processing methods via embeddings (i.e., vectors in
$\mathbb{R}^n$). Given the target labeler and a user-provided closeness function
over target labeler outputs, \sn produces embeddings for each unstructured data
record (e.g., frame of video), with the desideratum that close records have
close embeddings. \sn's required closeness function is often easy to specify,
e.g., that frames of a video with similar object types and object positions are
close (Section~\ref{sec:index-construction}).\footnote{ \sn can also be used
without training an embedding by using pre-trained embeddings, although query
performance will generally be better with its training method.  } \sn then uses
the embeddings and a small set of records annotated by the target labeler to answer
downstream queries.


Specifically, we propose a method of using \sn's embeddings and the labeled
records (i.e., cluster representatives) to generate proxy scores
\emph{automatically}, including for proxy-based aggregation, selection, and
limit query processing algorithms
(Section~\ref{sec:query-proc})~\cite{kang2019blazeit, kang2017noscope,
kang2020approximate, anderson2018predicate, lu2018accelerating}. \sn generates
per-record proxy scores by propagating annotations from the cluster
representatives to the unlabeled records. For example, we could assign an
unannotated frame the number of cars in the closest cluster representative for
counting cars. These scores can then used in query processing algorithms, such
as those in \blazeit, probabilistic predicates, Tahoma, etc. Moreover,
as the target labeler is executed over more data during query processing, we
can incrementally improve \sn's clustering, which improves performance (i.e.,
\sn's indexes support ``cracking''~\cite{idreos2007database}).

To understand \sn's performance, we provide a theoretical analysis of \sn and
downstream query accuracy (Section~\ref{sec:analysis}). We prove that queries
that are Lipschitz-continuous functions of the data will achieve \emph{exact}
results when using \sn (with sufficiently dense clustering) under 0 training
loss, and quantitative bounds when the training loss is not 0. \colora{Although
our assumptions are strong, our analysis provides statistically grounded
intuition for why TASTI can outperform baselines. We validate our intuition with
extensive experiments (Section~\ref{sec:eval}).}

We implemented \indname in a prototype system and evaluated it on four datasets,
including widely studied video datasets~\cite{kang2019blazeit,
kang2020approximate, xu2019vstore, jiang2018chameleon, canel2019scaling}, a text
dataset~\cite{zhongSeq2SQL2017}, and a speech dataset~\cite{ardila2019common}.
We integrated \sn into query processing algorithms for aggregation, selection,
and limit queries as proposed by prior work and executed these queries over the datasets. We show that \sn's
indexes require up to 10$\times$ \emph{fewer} labels from the target labeler to
construct than generating training data for per-query proxy model methods, as
\sn leverages redundancy in the datasets. Furthermore, \sn outperforms on query runtime across
all queries and datasets we evaluate on by up to 24$\times$ over previous
optimized systems.

In summary, our contributions are:
\begin{enumerate}[itemsep=0em,parsep=0em,topsep=0em,leftmargin=1.4em]
  \item We propose a method for constructing semantic indexes (\indname) for
  queries over unstructured data.
  \item We theoretically analyze \sn, providing a statistical understanding of
  our algorithm's performance.
  \item We evaluate \sn on five datasets and three query types,
  showing it can outperform state-of-the-art.
\end{enumerate}


\section{Overview and Example}
\label{sec:sys-overview}

\subsection{Background and Problem}
We first describe how target labelers and proxy scores are used in analytics systems
before describing an overview of \sn.

\colora{
Many analytics applications over unstructured data are powered by expensive DNNs
that extract structured data from these unstructured sources or human labelers
(e.g., ground-truth annotations for studying social or life sciences
\cite{kang2021accelerating}). We refer to these expensive DNNs and human
labelers as \emph{target labelers}.} These target labelers \emph{induce a schema}
over the extracted data. For example, object detection DNNs can extract
information about object types and positions from frames of a video. The schema
would contain columns for object type, object $(x,y)$-positions, and timestamps.
Unfortunately these high-quality target labelers, such as Mask R-CNN or BERT, can be
expensive and dominate query execution costs.

Thus, recent work attempts to accelerate queries with target labelers by using proxy
scores (e.g., \blazeit, \noscope, probabilistic predicates, Tahoma, \supg,
etc.). The most common method of producing proxy scores is to train a smaller
DNN (also called a ``specialized NN'' or ``proxy model'') that will produce a
proxy score per data record. These proxy models are typically trained to
approximate the output of the target labeler for the query at hand, e.g., a count
for an aggregation query, and can yield substantial query speedups.
Constructing the training data can be expensive as the training data
must reflect a wide range of potential queries.

We describe two examples of using proxy scores to accelerate queries, both of
which train a new proxy model per query. In both cases (and more generally), the
goal is to generate proxy scores that are highly correlated with the target
labeler outputs: these algorithms will adaptively improve with better proxy
scores.

\miniheadit{Approximate aggregation}
Suppose the user issues a query for the average number of cars per frame in a
video, as studied by \blazeit \cite{kang2019blazeit}. \blazeit takes as input an
error target and proxy scores.

To optimize this query, \blazeit trains a cheap proxy model whose output is the
predicted number of cars per frame using a sample of frames annotated by the
target labeler. This proxy model is then used to generate a
query-specific proxy score per frame. \blazeit then uses these scores as a
``control variate'' \cite{hammersley1964general} to reduce the variance in
estimation. Proxy scores that are more correlated with the true count will
result in faster query execution.

\miniheadit{Approximate selection}
Suppose the user issues a query to select 90\% of frames with cars with 95\%
probability of success, as studied by the ``recall target'' setting in \supg
\cite{kang2020approximate}. Specifically, \supg takes as input a target recall
and (in contrast to \blazeit) a fixed target labeler budget. \supg will train a
proxy model that estimate the probability of a record matching a predicate.

Given proxy scores, \supg will use importance
sampling and return a set of records achieving the recall target. Other recent
proxy-based systems accelerate selection queries without guarantees
\cite{anderson2018predicate, hsieh2018focus, lu2018accelerating,
bastani2020vaas}.

\subsection{\sn's Inputs, Outputs, and Goals}

\begin{figure*}[t!]
  \begin{subfigure}{0.48\textwidth}
    \includegraphics[width=0.99\textwidth]{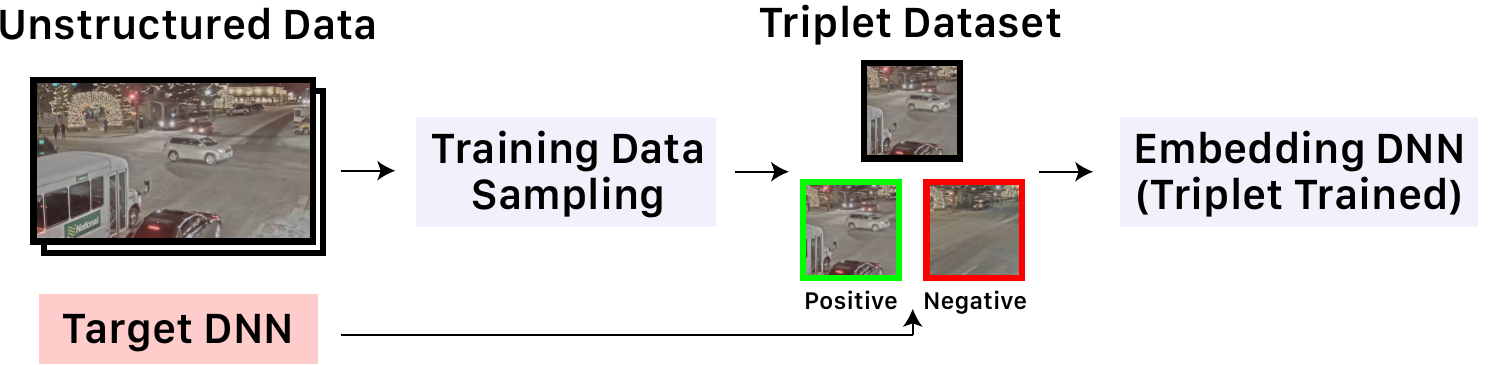}
    \caption{Overview of \sn's procedure for training embeddings. Training data
    is selected via the induced schema and a pre-trained embedding DNN. This
    data is then used to train an embedding DNN via the closeness function
    provided by the user.}
    \label{fig:sys-training}
    \vspace{1.5em}
  \end{subfigure}
  \hspace{0.15em}
  \begin{subfigure}{0.48\textwidth}
    \includegraphics[width=0.99\textwidth]{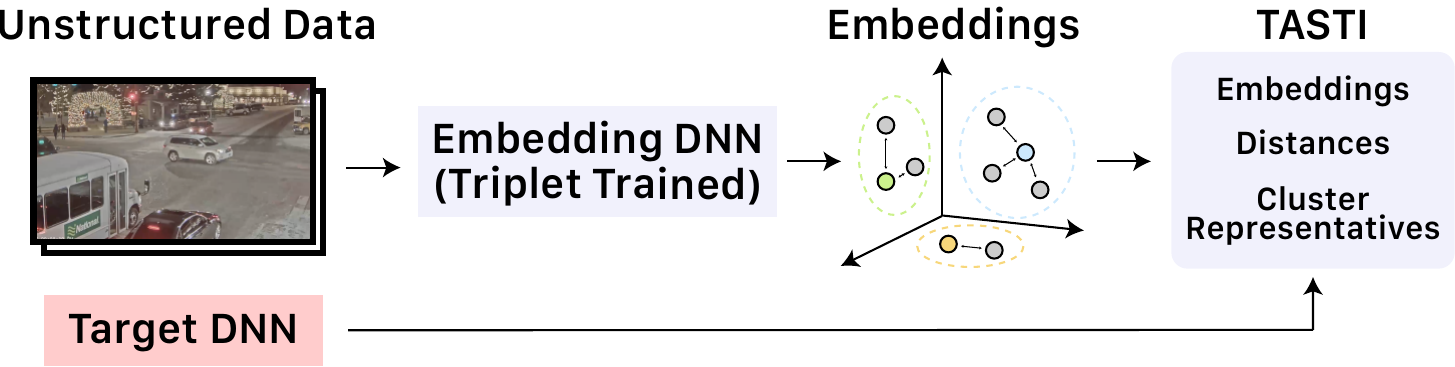}
    \caption{Overview of \sn's index construction procedure. \sn computes
    per-record embeddings, selects sample records to annotate (cluster
    representatives), and computes embedding distances from the unannotated
    records to the cluster representatives. \sn then stores the target labeler
    results, distances, and embeddings as its index.}
    \vspace{1em}
    \label{fig:sys-index-construction}
  \end{subfigure}
  \begin{subfigure}{0.98\textwidth}
    \includegraphics[width=0.98\textwidth]{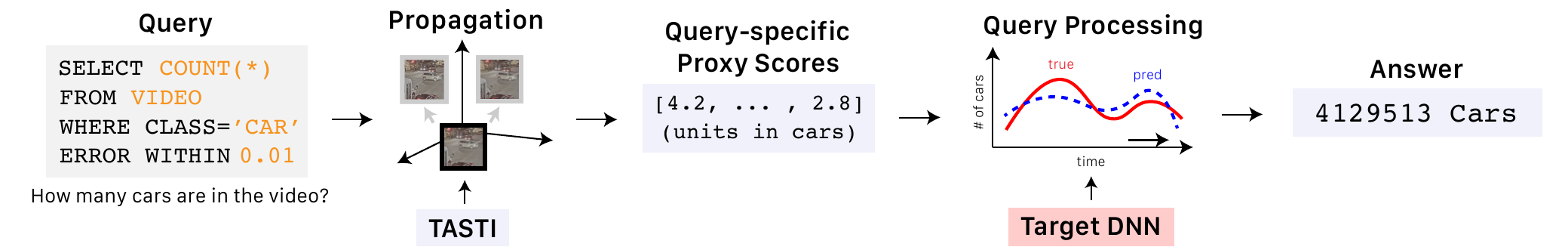}
    \caption{Overview of \sn's query processing. Given a query, \sn will compute
    exact results on the cluster representatives. It will then produce proxy
    scores on unannotated records by propagating the exact scores via embedding
    distance. These scores can then be used in existing downstream query
    processing algorithms based on proxy scores.}
    \label{fig:sys-query-proc}
  \end{subfigure}
  \vspace{-0.5em}
  \caption{\sn system overview.}
  \label{fig:sys-overview}
  \vspace{-0.5em}
\end{figure*}

\minihead{Overview}
As input, \sn takes a target labeler, an induced schema, a target labeler invocation
budget for index construction, a closeness function over the induced schema, and
a parameter $k$ that specifies how many distances to store for reach record. The
primary cost in index construction are the target labeler invocations. \sn will
produce an embedding-based index subject to the budget that can produce proxy
scores for a range of queries.

\sn's primary goal is to produce high quality proxy scores for query processing
algorithms, as with per-query proxy models, but \emph{without} training a new
model per query.

\colora{
\minihead{Supported queries}
We demonstrate how to generate proxy scores for selection queries, aggregation
queries, and limit queries \cite{kang2017noscope, kang2019blazeit,
kang2020approximate, anderson2018predicate, lu2018accelerating} with \sn. \sn
can be used with other queries requiring proxy scores. Since the initial draft,
other work has used \sn to support aggregation queries with predicates
\cite{kang2021accelerating}.

}

\subsection{Example}
Consider constructing an index for visual data, in which queries over object
types and positions are issued. In this case, the target labeler (e.g., Mask R-CNN)
takes an unstructured frame of video and returns a structured set of records
that contains fields about the positions and types of objects in the frame.
Consider two queries, one of which counts the number of cars per frame
(aggregation query as supported by \blazeit) and one that selects frames with
cars (selection query as supported by \noscope, probabilistic predicates,
\supg, etc.).
To understand how the index construction procedure and query processing works,
we describe the intuition below.

\minihead{Index construction}
\sn builds its index by training an embedding DNN for the input data and
then clustering results based on it. Ideally, semantically similar records are
grouped together, e.g., all frames with two cars might form one cluster, all
frames with one bike and one car might form a cluster, etc.

To train an embedding DNN, \sn requires a heuristic for
``close'' and ``far'' target labeler outputs, either as a Boolean function or as a
cutoff based on a continuous distance measure. One such heuristic for our video
application is to group frames with the same number of objects and similar
positions together. The grouping of ``close'' frames can be specified in
pseudocode as follows:
\begin{lstlisting}
def is_close(frame1: List[Box], frame2: List[Box])->bool:
    if len(frame1) != len(frame2):
        return False
    return all_boxes_close(frame1, frame2)
\end{lstlisting}
where \texttt{all\_boxes\_close} is a helper function that returns true if all boxes in
\texttt{frame1} have a corresponding ``close'' box in \texttt{frame2}. Given the
closeness function, \sn trains a low-cost embedding DNN via the triplet loss~\cite{weinberger2009distance},
which separates ``far'' frames. Then, \sn computes embeddings
over all frames of the video with the embedding DNN and select a set of frames to annotate with the
target labeler. It will then store the target labeler's outputs (object types and positions)
for each cluster representative. \sn uses the cluster representatives and
distances from unannotated frames' embeddings for downstream query processing.


\minihead{Query processing}
\sn can now be used to produce proxy scores for a range of downstream queries using existing proxy-based algorithms.
First, \sn generates exact scores on cluster
representatives, e.g., the exact count of the number of
cars from the cached target labeler outputs. Then, \sn propagates these scores
to the remainder of the records, e.g., producing approximate counts for the
unannotated records. We describe
the intuition behind two example queries: in both of these examples, \sn need
not train a new proxy model per query and can reuse its index.

\miniheadit{Approximate aggregation}
Consider the query of counting the average number of cars per frame. \sn computes the
query-specific proxy score as the distance-weighted average of the number of
cars in the $k$ closest annotated frames (see Section~\ref{sec:query-proc} for
pseudocode). This will produce an estimate of the number of cars in a given
unannotated frame. These scores can then be used by \blazeit's query processing
algorithm \cite{kang2019blazeit}.


\miniheadit{Approximate selection}
To estimate the probability of matching the predicate (i.e., query-specific proxy score) for \supg, \sn
will compute the weighted average as above, except that annotated frames that
contain a car receive a score of 1 and annotated frames that do not contain a
car receive a score of 0. These scores can then be used by \supg's selection
algorithm \cite{kang2020approximate}.


\vspace{0.5em}

We now discuss \sn's index construction method
(Section~\ref{sec:index-construction}) and \sn's query processing method
(Section~\ref{sec:query-proc}).

\section{Index Construction}
\label{sec:index-construction}

We describe how \sn constructs indexes, which can be used to produce high
quality proxy scores without the use of query-specific proxy models. Many
queries only require a low dimensional representation of data records to answer,
such as object types and positions (as opposed to raw pixels in a video).
Furthermore, in many applications, this low dimensional representation has a
natural closeness function, which can be directly used to construct high quality
proxy scores. \sn attempts to construct representations that reflect these
heuristics by grouping close records and separating far records.

We show a schematic of the training in Figure~\ref{fig:sys-training}, the
index construction in Figure~\ref{fig:sys-index-construction}, and the overall
algorithm in Algorithm~\ref{alg:index}. \sn's index
construction procedure consists of optionally training an embedding DNN via the
triplet loss, producing embeddings per record, selecting cluster
representatives, and computing statistics over the cluster representatives.

Throughout, we use the furthest point first (FPF) clustering algorithm
\cite{gonzalez1985clustering}. FPF iteratively chooses the furthest point from
the existing cluster representatives as the new representative. FPF performs
well in practice, is computationally efficient, and provides a 2-approximation
to the optimal maximum intra-cluster distance (which our analysis uses).

\subsection{Training the Embeddings}

\sn optionally trains a mapping between data records (e.g., frames
of a video) and semantic embeddings. The semantic embeddings have
the desideratum that data records that have similar extracted attributes are
close in embedding space, and vice versa for records that have dissimilar
extracted attributes. For example, consider queries over object type and
position. A frame with a single car in the upper left should be close to another
frame with a single car in the upper left, but far from a frame with two cars in
the bottom right.

We describe our training method via domain-specific triplet losses and show a
schematic in Figure~\ref{fig:sys-training}. We note that \sn's training
procedure is optional: pre-trained embeddings can be also be used for the index
if training is expensive.

\minihead{Domain-specific triplet loss}
To train the embedding DNN, \sn uses the triplet
loss~\cite{weinberger2009distance}. The triplet loss takes an anchor point, a
positive example (i.e., a close example), and a negative example (i.e., a far
example). It penalizes examples where the anchor point and the positive point
are further apart than the anchor point and the negative point (see
Section~\ref{sec:analysis}).

A key choice in using the triplet loss is selecting points that are ``close''
and those that are ``far.'' This choice is application specific, but many
applications have natural choices. For example, any frame of video with
different numbers of objects may be far (see Section~\ref{sec:sys-overview} for
pseudocode). Furthermore, frames with the same number of objects, but where the
objects are far apart may also considered far.

\begin{algorithm}[t!]
\caption{
  Pseudocode for \sn's index construction procedure. Given a dataset $X$,
  training points $N_1$, number of cluster representatives $N_2$, and min-$k$ to
  retain, \sn will construct the index as follows. FPF is the furthest point
  first algorithm, where the arguments are the embeddings and the number of
  points to select.
}
\label{alg:index}
\begin{algorithmic}
  \Function{Make \sn index}{$X$, $N_1$, $N_2$, $k$}
    \State PretrainedEmbeddings$[i]$ $\gets$ PretrainedModel$(X[i])$
    \State TrainingPoints $\gets$ FPF(PretrainedEmbeddings, $N_1$)
    \State TripletModel $\gets$ Finetune(TrainingPoints, PretrainedModel)
    \State Embeddings$[i]$ $\gets$ TripletModel$(X[i])$
    \State ClusterRepresentatives $\gets$ FPF(Embeddings, $N_2$)
    \State MinKDistances[i] $\gets$ ClosestKDistances($X[i]$, ClusterRepresentatives, $k$)
    \State \Return ClusterRepresentatives, MinKDistances
  \EndFunction
\end{algorithmic}
\end{algorithm}

\minihead{Training data selection (FPF mining)}
Training via the triplet loss requires invocations of the target labeler to
determine whether pairs of records are close or not. Due to the cost of the
target labeler, \sn must sample records to be selected for training; we assume the
user provides a budget of target labeler invocations. While \sn could randomly
sample data points, randomly sampled points may mostly sample redundant records
(e.g., majority of empty frames) and miss rare events. We empirically show that
randomly sampling training data results in embeddings that perform well on
average, but can perform poorly on rare events (Section~\ref{sec:eval}).

To produce embeddings that perform well across queries, we would ideally sample
a diverse set of data records. For example, suppose 80\% of a video were empty:
selecting frames at random would mostly sample empty frames. Selecting frames
with a variety of car numbers and positions would be more sample efficient.

When available, \sn uses a pre-trained DNN to select such diverse points. These
pre-trained DNNs are widely available, e.g., DNNs pre-trained on ImageNet
\cite{he2016deep} or on large text corpora (BERT) \cite{devlin2018bert}.
Pre-trained DNNs produce embeddings that are semantically meaningful, although
not adapted to the specific induced schema.

To produce training data that results in embeddings that perform well on rare
events, \sn performs the following selection procedure. First, \sn uses a
pre-trained DNN to generate embeddings over the data records. Then, \sn
executes the FPF algorithm to select the training data. \sn constructs triplets
from the training data via target labeler annotations. \sn will first bucket records
by the closeness function. To construct a triplet, \sn will sample two different
buckets at random: it will select the anchor and positive record at random from
the first bucket and a negative record at random from the second bucket.

\subsection{Clustering}
\sn produces clusters via the embedding DNN. As we describe in
Section~\ref{sec:query-proc}, \sn propagates annotations/scores from cluster
representatives to unannotated data records.

A key choice is which data records to select as cluster
representatives. Similar to selecting training data, \sn could select a set of
cluster representatives at random. While random sampling may do well on average
at query time, it may perform poorly on rare events (i.e., outliers).

To address this issue, \sn selects cluster representatives via FPF.  FPF chooses
points that are far apart in embedding space. Thus, if the embeddings are
semantically meaningful, then FPF will select data records that are diverse.
Finally, we mix a small fraction of random clusters, which helps ``average-case
performance'' queries.

\sn stores the distances of all embeddings to each cluster representative. As
we describe in Section~\ref{sec:query-proc}, \sn uses the $k$ nearest cluster
representatives for query processing.

\subsection{Cracking \sn Indexes}
In contrast to prior work, which can only share work between queries in an
ad-hoc manner, \sn's proxy scores will improve as queries are executed.
In particular, when any query executes the target labeler on a data record, \sn can cache the target
labeler result. The records over which the target labeler are executed can then be
added as new cluster representatives. Computing the distance to the new cluster
representative is computationally efficient and trivially parallelizable. We
note that this is a form of ``cracking'' \cite{idreos2007database}.

\subsection{Computational Performance}
\label{sec:index-perf-analysis}
Suppose there are $N$ data records, $D$ dimensions, $L$ training iterations, and
a total target labeler budget of $C$. Denote the costs of the target labeler, embedding
DNN, and distance computation as $c_T$, $c_E$, and $c_D$ respectively. The total
cost of index construction is $O(C \cdot c_T + L \cdot c_E + N \cdot c_E + N C D
\cdot c_D)$ assuming the cost of a training iteration is proportional to the
cost of the forward pass \cite{justus2018predicting}.

The ratio of these steps depends on the relative computational costs. In many
applications, the cost of embedding is less expensive than the cost of the
target labeler. For example, Mask R-CNN can execute as slow as at 3 fps,
compared to an embedding DNN which executes at 12,000 fps
\cite{kang2020jointly}. \colora{Furthermore, human labelers are orders of
magnitude more expensive than embedding DNNs (up to 100,000$\times$ more
expensive).}

\section{Query Processing with \sn}
\label{sec:query-proc}

How can \sn indexes accelerate query processing? We propose automatic
methods of construct query-specific proxy scores with \sn, which can then be
passed to existing proxy score-based algorithms. These query-specific proxy
scores are an approximation of the result of executing the target labeler on the
data records for the particular query. Consider an aggregation query counting
the average number of cars per frame \cite{kang2019blazeit}. The query-specific
proxy scores would be an estimate of the number of cars in a given frame.

Many downstream query processing algorithms only require proxy scores and the
target labeler. For example, selection without guarantees (e.g.,
binary detection) \cite{kang2017noscope, anderson2018predicate,
lu2018accelerating}, selection with statistical guarantees
\cite{kang2020approximate}, aggregation \cite{kang2019blazeit}, and limit
queries \cite{kang2019blazeit} only require query-specific proxy scores and the
target labeler. 

\sn provides a default method of taking the target labeler output and producing a
numeric score, which can support aggregation, selection and limit queries. Its
default method produces exact scores on the cluster representatives and
propagates using the distance-weighted average for numeric columns and
distance-weighted majority vote for categorical columns. If desired, a developer may also
implement custom functions to produce proxy scores for other queries. We
describe several examples of how proxy scores can be computed and used for
common query types, and then describe the interface for implementing custom
proxy scores. We show a schematic of the query processing procedure in Figure
\ref{fig:sys-query-proc}.


\subsection{Query Processing Examples}
\label{sec:qp-examples}
We provide examples of the query-specific scoring functions, score
propagation, and downstream query processing for several classes of queries
below.

\minihead{Approximate aggregation}
Consider the example of counting the average number of cars per frame, as
studied by \blazeit~\cite{kang2019blazeit}. The scoring function would take the
detected boxes in a frame and return the count of the boxes matching ``car,'' as
shown above.
For $k=1$, the query-specific proxy score would be the count for the nearest
cluster representative and for $k>1$, it would be the distanced-weighted mean
count of the nearest $k$ cluster representatives for a given frame.

The query-specific proxy scores can be used to answer the query with statistical
error bounds, e.g., used as a control variate by the \blazeit's query processing
algorithm. The scores could also be used to directly answer the query.

\minihead{Selection}
Consider a query that selects all frames of a video with a car, as studied by
prior work~\cite{kang2017noscope, anderson2018predicate, lu2018accelerating,
kang2020approximate}. The scoring function would take the detected boxes in a
frame and return 0 if there are no cars and 1 if there is a car in the frame.
The query-specific proxy score can be smoothed for $k>1$.

The query-specific proxy scores can be used as input to \supg, in which sampling
is used to achieve statistical guarantees on the recall or precision of selected
records~\cite{kang2020approximate}. These scores can also be used directly to
answer the query (i.e., return the records with value above some threshold,
either ad-hoc or computed over some validation set), as other systems
do~\cite{kang2017noscope, anderson2018predicate, lu2018accelerating}.

\minihead{Limit queries}
Consider a query that selects 10 frames containing at least 5
cars~\cite{kang2019blazeit}. Such queries are often used to manually study rare
events. In this case, the scoring function and query-specific proxy scores would
be the same as for aggregation. For limit queries, we generally recommend using
$k=1$, since this query is typically focused on ranking rare events.
The query processing algorithm will examine frames with the target labeler as
ordered by the query-specific proxy scores. The algorithm will terminate once
the requested number of frames is found.

\subsection{Custom Proxy Scores}

\sn has built-in functionality to compute and propagate scores for selection,
aggregation and limit queries using the methods described in the previous
section. In addition, developers may specify custom scores to extend TASTI to
supporting other queries.

\sloppypar{
The API for specifying scoring functions is as follows. Denote the type of
the output of the target labeler as \texttt{TargetLabelerutput} (e.g., a list of
bounding boxes) and the type of the score as \texttt{ScoreType} (e.g., a float).
Using Python typing, the developer would implement:
}
\begin{lstlisting}
def Score(target_output: TargetLabelerOutput) -> ScoreType
\end{lstlisting}
These functions can be implemented in few lines of code. We show the pseudocode
for the example above:
\begin{lstlisting}
def CountCarScore(boxes: Sequence[Boxes]) -> int:
    return len([box for box in boxes
                if box.object_type == 'car'])
\end{lstlisting}
Other queries, e.g., over object positions, can be implemented similarly with
few lines of code.

\subsection{Score Propagation}
Given the query-specific scoring functions, \sn will execute the scoring
functions on the cluster representatives (as the target labeler outputs are
available for these data records). In order to execute downstream query
processing, \sn must also materialize approximate scores for the remainder of
the data records.

To produce these query-specific proxy scores, \sn will propagate scores from the
cluster representatives to the unannotated records. The score for
each data record will be the inverse distance-weighted mean of the nearest $k$ cluster
representatives for numeric scores. For categorical scores, \sn will take the
distance-weighted majority vote. Since the distances to cluster representatives
are cached, this process is computationally efficient.
A developer may also implement a custom method of propagating scores. We show an
example of such a method in Section~\ref{sec:eval-fu} for limit queries.


\section{Theoretical Analysis}
\label{sec:analysis}

We present a statistical performance analysis of our methods to better
understand resulting query quality. Intuitively, if the original data records
have a metric structure and the triplet loss recovers this structure, we expect
downstream queries to behave well. Specifically, we provide guarantees on query
quality (typically accuracy) when using \sn directly. We show that accuracy for a natural class of
``smooth'' queries is directly connected to the triplet loss and the density of
clustering.
\colora{
While the assumptions in our analysis may not hold in practice, we conduct our
analysis to provide statistically grounded intuition for why \sn can outperform
baseline methods. We validate our intuition with extensive experiments
(Section~\ref{sec:eval}).

}

We formalize this intuition by analyzing how downstream queries behave under the
triplet loss. We specifically analyze the case where $k=1$, i.e., using a single
cluster representative in query processing.


\subsection{Notation and Preliminaries}
\minihead{Notation}
We define the set of data records as $\universe
:= \{ x_1, ..., x_N\}$, the scoring function $f(x_i): \universe \to \mathbb{R}$,
and the embedding function $\emb(x_i): \universe \to \mathbb{R}^d$.
Denote the cluster representatives as $R := \{ x_r : r \in \mathcal{R} \}
\subset \universe$ for some set $\mathcal{R} \subset \{1, ..., N\}$. Given this
set, we denote the representative mapping function as $c(x_i): \universe \to R$,
which maps a data record to the nearest cluster representative, and the
query-specific scores as $\hat{f}(x) := f(c(x))$.

Suppose there is a query-specific loss function $\lossq(x_i, y_i): \universe
\times \mathbb{R} \to \mathbb{R}$ where $y_i \in \mathbb{R}$ is the predicted
label. $\lossq$ will be used to evaluate the quality of $f$ and $\hat{f}$ as
$\lossq(x, f(x))$ and $\ell_Q(x, \hat{f}(x))$.

We define the per-example triplet loss as
{
\small
\[
\ell_T(x_a, x_p, x_n; \emb, m) := \max(0, m + |\emb(x_a) - \emb(x_p)| - |\emb(x_a) - \emb(x_n)|)
\]
}
where we omit $\emb$ and $m$ where clear. Define the ball of radius $M$ as
$B_M(x) = \{ x' : d(x, x') < M \}$ and its complement $\bar{B}_M$.  For random
variables $x_a \sim \universe$, $x_p \sim B_M(x_a)$, and $x_n \sim
\bar{B}_M(x_a)$ drawn uniformly from the sets, we define the population triplet
loss as
\begin{align}
    L(\emb; M, m) := \Exp_{x_a, x_p, x_n}[ \ell_T(x_a, x_p, x_n; \emb, m) ]
\end{align}
for some margin $m > 0$.

%

\minihead{Assumptions and properties}
We make the following assumptions. We first assume that there is a metric
$d(x_i, x_j)$ on $\universe$ and that $\universe$ is compact with metric $d$.
We further assume that $\lossq(x, y)$ is Lipshitz in $x$ and $y$ with constant
$K_Q/2$, in both arguments.

For both of our proofs, we assume the triplet loss is low and the cluster
representatives are dense enough under $\emb$. Low triplet loss controls the
quality of the embeddings with respect to the original metric $d$. The density
of the cluster representatives controls how close the unannotated records are
from the cluster representatives in the original space.

\minihead{Example}
Consider the video setting described in
Section~\ref{sec:sys-overview}. $\universe$ is the set of frames, $\emb$ is the
trained embedding DNN, and we use the metric induced by closeness function
also described in Section~\ref{sec:sys-overview}.
Consider the two queries: aggregation queries for the number of cars and
selecting frames of cars. For the aggregation query, $f$ maps frames to the
number of cars. For the selection query, $f$ maps frames with cars to 1 and
frames without cars to 0.

\subsection{Theorem Statements}

We defer all proofs to Appendix~\ref{sec:proofs}.

\minihead{Zero loss case}
To theoretically analyze our index and query processing algorithms, we first
consider the case where the embedding achieves zero triplet loss (we generalize
to non-zero loss below). We show the following positive result: using the
query-specific proxy scores in this setting will achieve bounded loss. In fact, for
$\lossq$ that are identically 0 (e.g., for the example above), \sn will achieve
\emph{exact} results.

We now state the main theorem for the zero-loss case.

\begin{thm}[Zero loss]
\label{thm:zero-loss}
Let $\emb$ be an embedding that achieves $L(\emb; M, m) = 0$ and $c$ be such
that $\max_{x \in \universe} |\emb(x) - \emb(c(x))| < m$. Then, the query
procedure will suffer an expected loss gap of at most 
\begin{align}
    \Exp[\lossq(x, \hat{f}(x))] \leq \Exp[\lossq(x, f(x))] + M \cdot K_Q.
\end{align}
\end{thm}

\minihead{Generalization to Non-zero Loss}
We generalize our analysis to the non-zero loss case below. We show that the
loss in queries is bounded by the triplet loss and several other natural
quantities.

\begin{thm}[Non-zero loss]
\label{thm:non-zero-loss}
Consider an embedding $\emb$ that achieves $L(\emb; M, m) = \alpha$ and a
clustering $c$ such that $\max_{x \in \universe}|\emb(x) - \emb(c(x))| < m$.
Assume that the query loss $\lossq$ is upper bounded by $C$. Then, query
procedure will suffer an expected loss gap of at most
{
\footnotesize
\begin{align}
\Exp[\lossq(x, \hat{f}(x))] \leq \Exp[\lossq(x, f(x))] +
    M \cdot K_Q +
    \frac{C \sup_x |\bar{B}_M(x)|}{m} \alpha.
\end{align}
}
\end{thm}

\subsection{Discussion}
We have shown that many classes of queries will have bounded loss (i.e.,
discrepancy from exact answers). However, we note that our analysis has several
limitations. First, \sn uses the nearest $k=5$ cluster representatives to generate the
query-specific proxy scores by default, not $k=1$ as used in our analysis. Second, the triplet loss may be large in practice.
Third, not all queries admit Lipschitz losses.  Nonetheless, we believe our
analysis provides intuition for why \sn outperforms even recent
state-of-the-art. We defer a more detailed analysis to future work.

\section{Evaluation}
\label{sec:eval}

We evaluated \sn on five real world datasets using three query types. We
describe the experimental setup and baselines. We then demonstrate that \sn's
index construction is cheaper than recent state-of-the-art executed end-to-end,
that \sn's proxy scores outperforms per-query proxies on all settings we
consider, that all components of \sn are required for performance, and that \sn
is not sensitive to hyperparameter settings. \colora{Our anonymized code is
available at \url{https://anonymous.4open.science/r/tasti-76FA}.}

\subsection{Experimental Setup}
\minihead{Datasets, target labelers, and triplet loss}
We considered three video datasets, a text dataset, and a speech dataset. We used
the \texttt{night-street}, \texttt{taipei}, and \texttt{amsterdam} videos as
used by \blazeit~\cite{kang2019blazeit}. The \texttt{night-street} dataset is widely used in
video analytics evaluations~\cite{kang2017noscope, kang2019blazeit,
canel2019scaling, xu2019vstore}. The \texttt{taipei} dataset has two object classes (car and
bus) and we use the \emph{same} set of embeddings for both. We used Mask R-CNN as
the target labeler and ResNet-18 as our embedding DNN. \colora{The closeness function
separates frames with objects that are far apart and frames with different
numbers of objects (when also considering object types).}

For the text dataset, we used a semantic parsing dataset~\cite{zhongSeq2SQL2017}.
The dataset consists of pairs of natural language questions and corresponding
SQL statements. We assumed the SQL statements are not known at query time and
must be annotated by crowd workers (i.e., that crowd workers are the target
labeler). We used BERT~\cite{devlin2018bert} for the embedding DNN. We considered
queries over SQL operators and number of predicates. \colora{The closeness function separates
questions over different SQL operators and number of predicates.}

For the speech dataset, we used the Common Voice dataset~\cite{ardila2019common}.
The dataset consists of short speech snippets. We assumed that the attributes of
speaker gender and age are not known at query time and must be annotated by
crowd workers. We used an audio ResNet-22~\cite{kong2020panns} for the embedding
DNN. \colora{The closeness function separates records by gender and discretized
age bucket.}

\minihead{Queries and metrics}
We evaluated \sn and per-query proxies on three classes of queries using three
recently proposed algorithms: aggregation, selection, and limit queries.

Our primary cost metric across all queries is the number of target labeler
invocations and also report end-to-end costs for certain experiments. \colora{We
use target labeler invocations as the primary metric for several reasons. First,
in many cases, the target labeler is actually a human labeler, particularly when
used in social or life sciences \cite{kang2021accelerating}. Second, the target
labelers we evaluate are thousands to hundreds of times more expensive than
query processing costs and proxy models~\cite{kang2020jointly}, and thus make up
the majority of query costs. In addition, this strictly benefits systems that
use per-query proxies which must be executed at query time: \sn does not train a
model per query.}

\miniheadit{Aggregation}
For aggregation queries, we queried for an approximate statistic of the target
labeler executed on the unstructured data records. We computed the average number of
objects per frame for the video datasets, the average number of predicates per
query for the WikiSQL dataset, and the fraction of male speakers in the Common
Voice dataset.

\colora{
For all settings, we used the EBS sampling as used by the \blazeit system
\cite{kang2019blazeit}, which provides guarantees on error. EBS sampling uses
the proxy scores to guide target labeler sampling. Better proxy scores will result
in fewer target labeler invocations. As such, we measured the number of target
labeler invocations (lower is better).

}

We additionally compare \sn to approximate aggregation without statistical
guarantees, which uses the proxy scores to answer queries directly (as used by
\blazeit).

\miniheadit{Selection}
For selection queries, we executed approximate selection queries with recall
targets (\supg queries \cite{kang2020approximate}). We selected for frames with
objects for video datasets, natural language questions that are parsed into
selection SQL statements for the WikiSQL dataset, and male speakers in the
Common Voice dataset.

\colora{
Given a target labeler budget, these queries return a set of records matching a
predicate with a given recall target with a given confidence level (e.g.,
``return 90\% of instances of cars with 95\% probability of success''): these
queries are useful in scientific applications or mission-critical settings
\cite{kang2020approximate}. In contrast to queries that do not provide
statistical guarantees, \supg guarantees the recall target with high
probability. Since recall \supg queries fix the number of target labeler
invocations, we measured the false positive rate (lower is better).

}

We additionally compare \sn to approximate selection without statistical
guarantees, which uses the proxy scores to answer queries directly. We slightly
modify the query processing algorithms of \noscope, Tahoma, and probabilistic
predicates to directly use proxy scores and use the accuracy metric of F1 score.

\miniheadit{Limit}
For limit queries, we used the ranking algorithm proposed by
\citet{kang2019blazeit}. This ranking algorithm examines data records that are
likely to match the predicate of interest in descending order by the proxy
score. Proxy scores that have high recall for given number of records will
perform better. As such, we measured the number of target labeler invocations (lower
is better).

\minihead{Methods evaluated}
We used the query processing methods above and use the per-query proxies as used
in \citet{kang2019blazeit} (aggregation and limit queries) and
\citet{kang2020approximate} (selection queries). We use the exact proxy models
for the video datasets (a ``tiny ResNet''), logistic regression over FastText
embeddings~\cite{bojanowski2017enriching} for the WikiSQL dataset, and a smaller
CNN (CNN-10)~\cite{kong2020panns} for the Common Voice dataset. FastText
embeddings are less expensive than BERT embeddings.

Throughout, we refer to \sn when using a pre-trained DNN as the embedding DNN as
``\sn-\textsc{PT}'' (pre-trained) and \sn when using a triplet-loss trained
embedding DNN as ``\sn-\textsc{T}'' (trained). We demonstrate that
\sn-\textsc{T} generally outperforms \sn-\textsc{PT}.

\minihead{Hardware and timing}
We evaluated \sn on a private server with a single NVIDIA V100 GPU, 2 Intel Xeon
Gold 6132 CPUs (56 hyperthreads), and 504GB of memory. In contrast to prior
work, we timed end-to-end query processing times for \sn, \emph{including} the
video loading and embedding DNN execution times, which is excluded in prior
work~\cite{kang2019blazeit}.

Due to the large cost of executing the target labeler, we simulated its
execution by caching target labeler results and computing the average
execution time for the target labeler. For baselines, we only timed the target
labeler computation and exclude the computational cost of proxy models, which strictly
improves the baselines. We excluded the cost of query
processing~\cite{kang2019blazeit, kang2020approximate} as it is negligible in
all cases. Namely, the query processing is over orders of magnitude less
expensive than target labeler invocation for all queries we consider.

\subsection{Index Construction Performance}

\begin{figure}
  \includegraphics[width=\columnwidth]{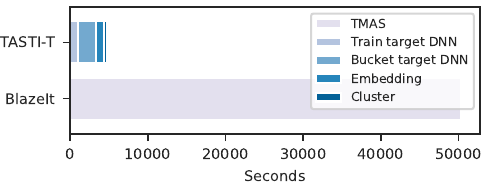}
  \caption{Breakdown of time to construct indexes for \sn and for \blazeit on
  the \texttt{night-street} dataset. The \blazeit index is the ``target-model
  annotated set'' (TMAS)~\cite{kang2019blazeit}. Similar results hold for other
  datasets.}
  \label{fig:eval-index-time}
\end{figure}

 \begin{figure}
   \includegraphics[width=\columnwidth]{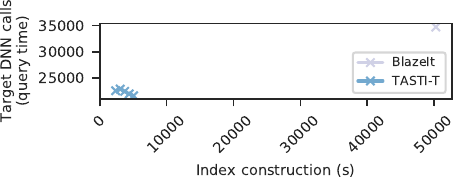}
   \caption{Index construction time vs performance of \sn and \blazeit for
   aggregation queries on the \texttt{night-street} dataset. Similar results hold
   for other datasets.}
   \label{fig:eval-index-vs-perf}
 \end{figure}

To understand the index construction performance, we measured the wall clock
time to construct \sn indexes. We compared to \blazeit, which constructs indexes
by executing the target labeler on a subset of the data (referred to as the ``TMAS''
\cite{kang2019blazeit}). For \blazeit, we only considered the cost of
constructing the TMAS. For \sn, we measured the full index construction time,
including the embedding DNN training and distance computation times. We computed
the construction times on the \texttt{night-street} dataset; similar results
hold for other datasets.

We show the breakdown of index construction time for \sn and \blazeit in
Figure~\ref{fig:eval-index-time} using the parameters in
Section~\ref{sec:eval-fu}. \sn requires far fewer target labeler invocations for
index construction, so is substantially faster than \blazeit.
We additionally show the index construction time vs performance for \blazeit and
a range of parameters for \sn (Figure~\ref{fig:eval-index-vs-perf}). \sn can
outperform or match \blazeit performance with up to 10$\times$ less expensive
index construction times.

\subsection{End-to-end Performance}
\label{sec:eval-fu}

We show that \sn outperforms recent state-of-the-art per-query proxy methods for
approximate aggregation, selection with guarantees, and limit queries. \colora{For all
video datasets in this section, we used 3,000 training records, 7,000 cluster
representatives, and an embedding size of 128. To show the generality of \sn, we
used a single set of embeddings/distances for both \texttt{taipei} classes. For
the WikiSQL and Common Voice datasets, we used 500 training examples and 500
cluster representatives. We measured the query processing costs or query accuracy
in this section.}
%

\begin{figure}[t!]
  \includegraphics[width=\columnwidth]{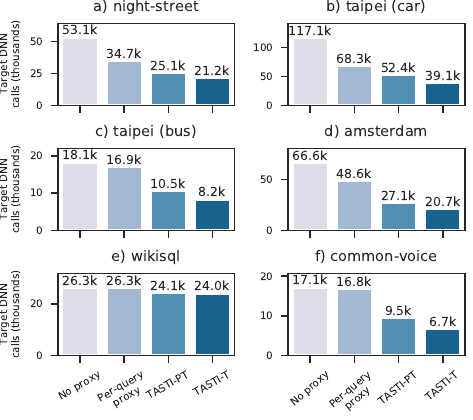}
  \caption{Number of target labeler invocations for baselines and \sn for
  \colora{approximate aggregation queries} (lower is better). As shown, \sn
  outperforms baselines in all cases, including prior, per-query proxy
  state-of-the-art by up to 2$\times$. \colora{All methods achieved the target
  accuracy.}}
  \label{fig:eval-fu-agg}
\end{figure}

\minihead{Approximate aggregation}
For approximate aggregation queries, we compared \sn to using no proxy (random
sampling) and an ad-hoc trained proxy model. We used the exact experimental
setup as \blazeit~\cite{kang2019blazeit} for video datasets, which targeted an
error of 0.01 and a success probability of 95\%. We aggregated over the average
number of objects per frame for all video datasets (cars or buses), the number
of clauses per statement in the WikiSQL dataset, and the fraction of male
speakers in the Common Voice dataset.

As shown in Figure~\ref{fig:eval-fu-agg}, \sn outperforms for aggregation
queries on all datasets. In particular, \sn outperform state-of-the-art
per-query proxies for aggregation queries (\blazeit) by up to 2$\times$ with
less expensive index construction costs. Further, \sn outperforms no proxy by up
to 3$\times$.

\sn's improved performance comes from better query-specific proxy scores
($\rho^2$ of 0.91 vs 0.55). As the correlation of the proxy scores with the
target labeler increases, the control variates variance decreases. Reduced variance
results in fewer samples, as the EBS stopping algorithm is adaptive with the
variance.

\begin{figure}[t!]
  \includegraphics[width=\columnwidth]{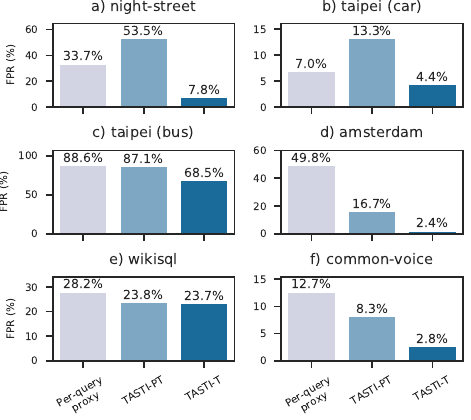}
  \caption{False positive rate for recall-target \supg queries (lower is
  better). We show the performance of baselines and \sn. As shown, \sn
  outperforms baselines in all cases.
  }
  \label{fig:eval-fu-supg}
\end{figure}

\minihead{Selection}
For selection queries with statistical guarantees (\supg queries), we compared
\sn to using an ad-hoc trained proxy model (standard random sampling is not
appropriate for \supg queries). We used the exact same experimental setup as in
\supg~\cite{kang2020approximate} for the video datasets.  For all queries, we
used a recall target of 90\% with a confidence of 95\%, as used
in~\cite{kang2020approximate}. We search for cars or buses in the video
datasets, star operators for WikiSQL, and male speakers in the Common Voice
dataset.

As shown in Figure~\ref{fig:eval-fu-supg}, \sn outperforms on all datasets. In
particular, \sn can improve the false positive rate by almost 21$\times$ over
recent state-of-the-art. We further show that the triplet training improves
performance.
As with aggregation queries, \sn's improved performance comes from better
query-specific proxy scores ($\rho^2$ of 0.90 vs 0.79).

\begin{figure}[t!]
  \includegraphics[width=\columnwidth]{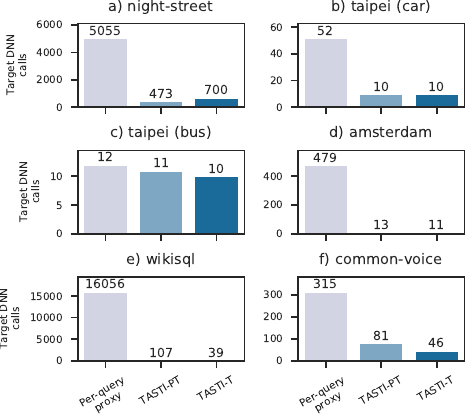}
  \caption{Number of target labeler invocations for baselines and \sn for
  limit queries (lower is better). \sn outperforms
  baselines in all cases, including prior state-of-the-art by up to 34$\times$.}
  \label{fig:eval-fu-limit}
\end{figure}

\minihead{Limit queries}
For limit queries, we used the ranking algorithm proposed by
\blazeit~\cite{kang2019blazeit}. We use the exact same experimental setup as
\blazeit for the video datasets (including the query configurations, e.g.,
number of objects, etc.). For limit queries, we use a custom scoring function
which is the regular scoring function with $k=1$ and ties broken by distance to
the cluster representatives.

Figure~\ref{fig:eval-fu-limit} shows \sn outperforms on all datasets. \sn
can improve performance by up to 24$\times$ compared to recent state-of-the-art.
As we demonstrate, \sn's FPF mining and FPF clustering are critical for
performance when searching for rare events (Section~\ref{sec:eval-factor},
Figure \ref{fig:eval-factor} and \ref{fig:eval-lesion}). The FPF algorithm
naturally produces clusters that are far apart, which is beneficial when
searching for rare events.

\begin{table}
\small
\begin{tabular}{lllll}
Target & \sn        & \sn         & Uniform    & Exhaustive \\
       & (no index) & (all costs) & (no proxy) & \\
\hline
Human labeler & \textbf{\$1,482} & \$1,972 & \$3,717  & \$68,116 \\
Mask R-CNN    & \textbf{7,060 s} & 9,474 s & 17,702 s & 324,362 s \\
SSD           & \textbf{141 s}   & 269 s   & 354 s    & 6,487 s
\end{tabular}
\caption{\colora{Query costs for \sn (when amortizing the cost of constructing
the index), \sn (including the cost of the index), uniform sampling, and
exhaustive labeling for answering an approximate aggregation query on the
\texttt{night-street} dataset. \sn, both with and without the indexing costs,
outperforms in all cases.}}
\label{table:query-costs}
\end{table}

\colora{
\minihead{Comparison without proxies (aggregation)}
To illustrate the performance of \sn, we also compare total costs of uniform
sampling with no index and to exhaustive labeling. In these experiments, we
compare against three different target labelers: human labelers, Mask R-CNN, and SSD (an
inexpensive object detection method). Different applications require different
levels of accuracy: the human labeler is the most accurate, followed by Mask
R-CNN, and finally SSD. Importantly, we note that SSD is over 2$\times$ less
accurate than Mask R-CNN (50.2 vs 23.0 mAP), so can result in inaccurate queries
when used as the target labeler.

We show the cost of \sn (index cost amortized), \sn (including index cost),
uniform sampling, and exhaustive labeling in Table~\ref{table:query-costs} for
the three different targets for aggregation on the \texttt{night-street}
dataset. We use the same approximate aggregation query as above, targeting an
error of 0.01 and a success probability of 95\%. As shown, the cost when using
\sn can be up to 46$\times$ cheaper to answer aggregation queries. Importantly,
\sn is cheaper in all cases \emph{when including the cost of building the index}
for answering this query.


Finally, we note that cheaper models are less accurate: SSD results in a 33\%
error compared to the more accurate Mask R-CNN. Matching the accuracy of TASTI
with Mask R-CNN as the target labeler takes 30.4 s to execute, which is cheaper
than exhaustively executing SSD.

}

\subsection{New Queries}
\label{sec:eval-new-queries}

\begin{figure}[t!]
  \includegraphics[width=\columnwidth]{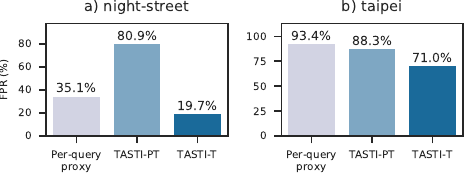}
  \caption{\supg queries for selecting objects of interest on the left hand side
  of the frame. This query violates the Lipschitz condition, but \sn still
  outperforms baselines.}
  \label{fig:eval-lhs-supg}
\end{figure}

\begin{figure}[t!]
  \includegraphics[width=\columnwidth]{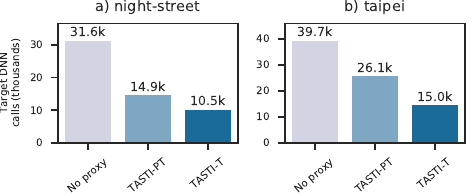}
  \caption{Aggregation query for the average $x$ position of objects in frames
  of a video. Recent state-of-the-art is not well suited for this query as
  regression can be difficult for proxy models. In contrast, \sn performs
  well on these queries.}
  \label{fig:eval-avg-pos}
\end{figure}

In addition to the queries above, we demonstrated that \sn can be used to
efficiently answer queries that prior work is not well suited for. We considered
two queries over positions of objects in video. In particular, these tasks
require modified data preprocessing or losses for proxy models, but \sn can
naturally produce proxy scores for both tasks.

\minihead{Selecting objects by position}
We considered the query of selecting objects in the left hand side of the video,
as measured by the average x-position of the bounding box. We compared \sn to to
training a proxy model by extending \supg and to \sn without triplet training.
Results are shown in Figure~\ref{fig:eval-lhs-supg}.

Prior proxy models were not designed to take position into account, which may
explain their poor performance: there is a sharp discontinuity for labels in the
center of the frame. Learning the boundary in the frame may require large
amounts of training data. In contrast, \sn performs well on as it uses the
information from the target labeler, despite the query violating our assumptions in
the theoretical analysis. As shown, \sn outperforms both baselines.

\minihead{Average position}
We consider the query of computing the average position of objects in frames of
video (specifically the x-coordinate). We compare \sn to random sampling (no
proxy) and to \sn without triplet training. We attempted to train a \blazeit
proxy model by regressing the output to the average position but were unable to
train a model that outperformed random sampling. \blazeit was not configured for
such queries and that we are unaware of work on proxy models for pure
regression. We show results in Figure~\ref{fig:eval-avg-pos}.
As shown, \sn outperforms random sampling by up to 3$\times$, without having to
implement custom training code for a new proxy model.


\begin{table}
\small
\begin{tabular}{llll}
  Dataset & Method & Query & Quality metric \\ \hline \hline
  \texttt{night-street} & \sn      & Agg. & \textbf{3.3\%} \\
  \texttt{night-street} & \blazeit & Agg. & 4.4\% \\
  \hline
  \texttt{night-street} & \sn      & Selection & \textbf{5.5} \\
  \texttt{night-street} & \noscope & Selection & 14.9
\end{tabular}
\caption{
Performance of \sn and baselines on queries without statistical guarantees
(lower is better). The quality metrics are percent error and 100 - F1 score for
aggregation and selection queries respectively. \sn outperforms on all settings
we considered.
}
\label{table:eval-no-guarantees}
\end{table}

\subsection{Queries Without Guarantees}
\label{sec:eval-no-guarantee}
In addition to queries with statistical guarantees, we executed aggregation and
selection queries without statistical guarantees. For aggregation queries, we
used the proxy score to directly compute the statistic of interest and measured
the percent error from the ground truth. For selection queries, we used the
proxy score to select records above some threshold. As some selection queries
are class imbalanced, we measured 100 - F1 score (so lower is better).

We show results for \sn and for proxy model-based baselines in
Table~\ref{table:eval-no-guarantees}. As shown, \sn outperforms on quality
metrics for all settings we considered, indicating that \sn's proxy scores are
higher quality.

\begin{table}
\small
\begin{tabular}{llll}
  Dataset & 1st query & 2nd query & Quality metric \\ \hline \hline
  \texttt{night-street} & Agg.   & \supg & 4.9\% (8.6\%) \\
  \texttt{taipei}       & Agg.   & \supg & 40.1\% (55.9\%) \\
  \hline
  \texttt{night-street} & \supg  & Agg.  & 18.9k (21.2k) \\
  \texttt{taipei}       & \supg  & Agg.  & 34.6k (39.1k)
\end{tabular}
\caption{
Performance of \sn after cracking. We measured query performance of a
\supg/aggregation query after cracking (false positive rate and number of target
labeler invocations, lower is better for both). Results after cracking are shown
along with results before cracking in parentheses. \sn improves results in all
settings we tested.
}
\label{table:eval-cracking}
\end{table}

\subsection{Cracking}
\label{sec:eval-cracking}

We further demonstrated that \sn's indexes can be
``cracked''~\cite{idreos2007database}. To show this, we executed an aggregation
query followed by a \supg query and vice versa. We used the target labeler
annotations from the first query to improve \sn's index before executing the
second query. We use the same quality/runtime metrics as for queries with
statistical guarantees.

As shown in Table~\ref{table:eval-cracking}, \sn improves in performance for
both queries. In particular, \sn can improve (decrease) the false positive rate
for \supg queries by up to 1.7$\times$ after repeated queries.

\subsection{Factor Analysis and Lesion Study}
\label{sec:eval-factor}

We investigated whether all of \sn's components contributes to performance. We
find that all components of \sn (triplet loss, FPF mining, and FPF clustering)
are critical to performance.

\begin{figure}
  \includegraphics[width=\columnwidth]{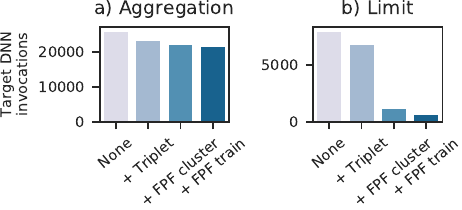}
  \caption{Factor analysis, in which optimizations are added in in sequence. As
  shown, all optimizations improve performance for aggregation queries. For
  limit queries, FPF training and clustering are required for triplet training to
  improve performance.}
  \label{fig:eval-factor}
\end{figure}

\minihead{Factor analysis}
We first performed a factor analysis, in which we began with no optimizations
and added the triplet loss, FPF mining, and FPF clustering in turn. For brevity,
we show results for the \texttt{night-street} dataset for aggregation and limit
queries.  Aggregation queries highlight ``average-case'' performance and limit
queries highlight ``rare-event'' performance. We choose the
\texttt{night-street} dataset as
it has been widely studied in visual analytics \cite{kang2017noscope,
kang2019blazeit, kang2020approximate, canel2019scaling, xu2019vstore}; other
datasets have similar behaviors.

As shown in Figure~\ref{fig:eval-factor}, all optimizations help performance. In
particular, FPF clustering substantially improves limit query performance, as it
selects frames that are semantically distinct.

\begin{figure}
  \includegraphics[width=\columnwidth]{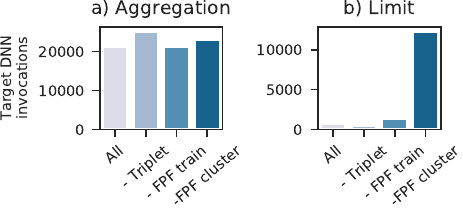}
  \caption{Lesion study, in which optimizations are removed individually (they
  are not removed cumulatively). As shown, all optimizations improve
  performance. Lower is better for both aggregation and limit queries.}
  \label{fig:eval-lesion}
\end{figure}

\minihead{Lesion study}
We then performed a lesion study, in which we start with all optimizations, and
remove each optimization individually (triplet loss, FPF mining, and FPF
clustering). As with the factor analysis, we show results for the
\texttt{night-street} dataset for aggregation and limit queries; other datasets have
similar behaviors.

We show results in Figure~\ref{fig:eval-lesion}. As shown, triplet training
significantly improves aggregation performance. Furthermore, FPF clustering is
critical for limit query performance.

\subsection{Sensitivity Analysis}
\label{sec:eval-sensitivity}

We investigated whether \sn is sensitive to hyperparameters by varying the
number of training examples, number of cluster representatives, and embedding
size. As we show, \sn outperforms baselines on a wide range of parameter
settings, demonstrating that hyperparameters are not difficult to select.

\begin{figure}
  \includegraphics[width=\columnwidth]{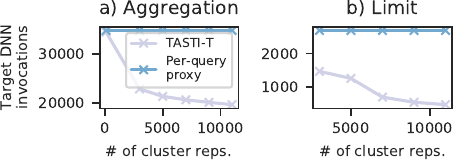}
  \caption{Number of cluster representatives vs performance on aggregation and
  limit queries on the \texttt{night-street} dataset. As shown, \sn outperforms
  baselines on a range of parameter settings.}
  \label{fig:eval-nb-buckets}
\end{figure}

\minihead{Number of buckets}
A critical parameter that determines \sn performance is the number of buckets in
the index. To understand the effect of the number of buckets on performance, we
vary the number of buckets and measured performance on aggregation and limit
queries on the \texttt{night-street} dataset. We used 3,000, 5,000, 7,000, 9,000,
and 11,000 buckets for both queries. For aggregation queries, we additionally
used 50 buckets.

As shown in Figure~\ref{fig:eval-nb-buckets}, \sn performance improves as the
number of buckets increases. For aggregation queries, \sn outperforms with as
few as 50 buckets. For limit queries, \sn outperforms with 5,000 buckets. We
note that this setting still corresponds to index construction over 10$\times$
less expensive than the baseline.

\begin{figure}
  \includegraphics[width=\columnwidth]{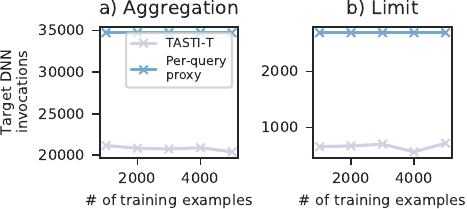}
  \caption{Number of training examples vs performance on aggregation and limit
  queries on the \texttt{night-street} dataset. As shown, \sn outperforms
  baselines on a range of parameter settings.}
  \label{fig:eval-nb-train}
\end{figure}

\minihead{Number of training examples}
To understand how the number of training examples affects the performance of
\sn, we used 1,000, 2,000, 3,000, 4,000, and 5,000 training examples. We
measured performance on aggregation and limit queries on the
\texttt{night-street} dataset.
As shown in Figure~\ref{fig:eval-nb-train}, the performance of \sn does not
significantly change with the number of training examples. \sn outperforms
baselines across all settings we consider.

\begin{figure}
  \includegraphics[width=\columnwidth]{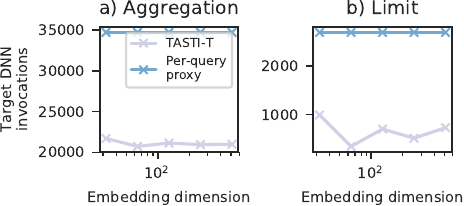}
  \caption{Embedding size vs performance on aggregation and limit queries on the
  \texttt{night-street} dataset. \sn outperforms baselines on a range of
  parameter settings.}
  \label{fig:eval-embed-size}
\end{figure}

\minihead{Embedding size}
To understand how the embedding size affects performance, we varied
the embedding size and measured performance on aggregation and limit queries on
the \texttt{night-street} dataset. We used embedding sizes of 32, 64, 128, 256,
and 512.
We show results in Figure~\ref{fig:eval-embed-size}. As shown, \sn outperforms
per-query proxies across a range of embedding sizes for aggregation and limit
queries.


\section{Related Work}
\label{sec:rel-work}

\minihead{DNN-based queries}
Recent work in the database and systems community has focused on accelerating DNN-based
queries. Many systems have been developed to accelerate certain classes of
queries, including selection without statistical guarantees
\cite{kang2017noscope, lu2018accelerating, anderson2018predicate,
hsieh2018focus}, selection with statistical guarantees
\cite{kang2020approximate}, aggregation queries, limit
queries~\cite{kang2019blazeit}, tracking queries \cite{bastani2020miris}, and 
other queries. These systems reduce the cost of expensive target
labeler, often by using cheap, query-specific proxy models. In this work, we
propose a general index to accelerate many such queries over the schema induced
by the target labeler. We leverage many of the downstream query processing
techniques used in this prior work.

Other work assumes that the target labeler is not expensive to execute or that
extracting bounding boxes is not expensive \cite{zhang2019panorama,
fu2019rekall}. We have found that many applications require accurate and
expensive target labelers, so we focus on reducing executing the target labeler wherever
possible.

\minihead{Structured data indexes}
There is a long history in the database literature of indexes for structured
data~\cite{ullman1984principles}. These indexes generally are used to
accelerated lookups on certain columns. Techniques range from tree-like
structures \cite{bayer1970organization, hellerstein1995generalized,
comer1979ubiquitous} to hash tables~\cite{garcia2008database}. However, these
indexes assume that the data is present in a structured format, which is not the
case for the data we consider.

\minihead{Unstructured data indexes}
The analytics community has also long studied indexes for unstructured data.
Many of these indexing methods are modality-specific, such as indexes for
spatial data~\cite{guttman1984r, ciaccia1997m}, time
series~\cite{agrawal1993efficient, faloutsos1994fast}, and low-level visual
features~\cite{flickner1995query, smith1997visualseek}. Other indexes accelerate
KNN search in possibly high dimensions, when the distances are
meaningful~\cite{yu2001indexing, jagadish2005idistance}.
Work in retrieval has used indexed embeddings to accelerate search for
semantically similar items, in particular for visual
data~\cite{babenko2014neural, lin2015deep, zhang2019panorama}. While this work
also accelerates queries over unstructured data, our work differs in focusing
on constructing proxy scores to address unique challenges when queries require
executing expensive target labelers.



\minihead{Coresets}
Several communities, including the theory and deep learning communities, have
considered coresets~\cite{agarwal2005geometric}, which are concise summaries of
data. They have been used for nearest neighbor searches~\cite{har2004coresets},
streaming data~\cite{braverman2016new}, active learning~\cite{sener2017active},
and other applications. We are unaware of work that trains embedding DNNs as an
index for proxy score generation.


\section{Conclusion}
\label{sec:conclusion}

To reduce the cost of queries using expensive target labelers, we introduce a method
of constructing indexes for unstructured data. \indname
relies on the key property that many queries only require access to target
labelers outputs, which are often highly redundant. \sn uses an embedding DNN and
target labeler annotated cluster representatives as its index, which allows for more
accurate and generalizable proxy scores across a range of query types. We
theoretically analyze \sn to understand its statistical accuracy. We show that
these indexes can be constructed up to 10$\times$ more efficiently than recent
work. We further show that they can be used to answer queries up to 24$\times$
more efficiently than recent state-of-the-art.

\appendix

\section{Proofs of Lemmas and Theorems}
\label{sec:proofs}


\begin{lemma}
\label{lemma:dist}
If $\ell_T(x_a, x_p, x_n) = 0$ for $x_a \in \universe, x_p \in B_M(x_a), x_n \in
\bar{B}_M(x_a)$, then for all $x_i, x_r$ such that $|\emb(x_i) - \emb(x_r)| < m$
we have $d(x_i, x_r) < M$.
\end{lemma}

\begin{proof}
To prove the lemma, we will show the contra-positive: if $d(x_i, x_r) \geq M$
implies that $|\emb(x_i) - \emb(x_r)| \geq m$, we have our result.

Let $x_a = x_i, x_n = x_r, x_p \in B_M(x_i)$. $B_M(x_i)$ must be nonempty since
$x_i \in B_M(x_i)$. This implies in inequality:
%
{
\small
\begin{align*}
    0  &\geq m + |\emb(x_a) - \emb(x_p)| - |\emb(x_a) - \emb(x_n)| \\
    |\emb(x_a) - \emb(x_n)| &\geq m.
\end{align*}
}
\end{proof}

\begin{proof}[Proof of Theorem \ref{thm:zero-loss}]
Since the triplet loss is bounded below by 0, $L(\emb; M, m) = 0$ implies that
$\ell_T(x_a, x_p, x_n) = 0$ for any $x_a, x_n$ such that $d(x_a, x_n) > M$,
since $\ell_T$ is bounded below by 0. By Lemma \ref{lemma:dist} with $x_i = x$
and $x_r = c(x)$, and since the maximum intra-cluster instance is $m$,
\[
d(x, c(x)) < M
\]
for all $x \in \universe$.

Then for every $x$:
\begin{align}
    |\lossq(x, f(x)) &- \lossq(x, \hat{f}(x))| \\
    &\leq |\lossq(x, f(x)) - \lossq(c(x), f(c(x)))| + \nonumber \\
    & \;\;\;\;\;  |\lossq(c(x), f(c(x))) - \lossq(x, f(c(x)))| \\
    &\leq M \cdot K_Q
\end{align}
This follows by the definition of $\hat{f}$, the Lipschitz condition of
$\lossq$, and the non-negativity of $\lossq$.

The proof follows from taking expectations.
\end{proof}


\begin{lemma}
\label{lemma:prob-bound}
{
\small
\begin{align*}
  \mathbb{P} \left[\inf_{x' \in \bar{B}_M(x)} |\emb(x) - \emb(x')| \leq |\emb(x) - \emb(x_p)| \right] \geq \\
  \mathbb{P}[d(x, c(x)) > M]
\end{align*}
}
for any distribution of $x_p$ such that the condition distribution of $x_p$ on
$x$ has support on $B_M(x)$.
\end{lemma}

\begin{proof}
Recall that $c(x) := \argmin_{x_r \in \mathcal{R}} |\emb(x) - \emb(x_r)|$ and that $d(x, c(x)) >
M$ implies $c(x) \in \bar{B}_M(x)$. Then, we have that $\inf_{x' \in
\bar{B}_M(x)} |\emb(x) - \emb(x')| \leq |\emb(x) - \emb(x_p)|$ for all $x_p \in
B_M(x)$. This gives us the lemma.
\end{proof}

\begin{lemma}
\label{lemma:triplet-bound}
\begin{align*}
\frac{1}{m} \ell_T(x, x_p, x_n) \geq
\mathds{1}_{|\emb(x) - \emb(x_n)| \leq |\emb(x) - \emb(x_p)|}
\end{align*}
for any $x_n \in \bar{B}_M(x), x_p \in B_M(x)$.
\end{lemma}

\begin{proof}
\begin{align*}
\frac{1}{m} \ell_T(x, x_p, x_n) &=
    \frac{1}{m} \cdot \max(0, m + |\emb(x) - \emb(x_p)| - |\emb(x) - \emb(x_n)|) \\
&= \max \left(0, 1 - \left( \frac{|\emb(x) - \emb(x_n)| - |\phi(x) - \emb(x_p)|}{m} \right) \right) \\
&\geq \mathds{1}_{|\emb(x) - \emb(x_n)| \leq |\emb(x) - \emb(x_p)|}.
\end{align*}
which follows from the hinge dominating the indicator.

\end{proof}

\begin{proof}[Proof of Theorem \ref{thm:non-zero-loss}]
Consider the indicators $\mathds{1}_{d(x, c(x)) \leq M}$ and
its complement $\mathds{1}_{d(x, c(x)) > M}$.

We analyze
\begin{align*}
&\Exp[\lossq(x, \hat{f}(x))] = \\
  &\Exp[\lossq(x, \hat{f}(x)) \cdot \mathds{1}_{d(x, c(x)) \leq M}] +
  \Exp[\lossq(x, \hat{f}(x)) \cdot \mathds{1}_{d(x, c(x)) > M}]
\end{align*}

By Theorem \ref{thm:zero-loss} and that expectations of indicators are bounded above by 1, we have that
\[
\Exp[\lossq(x, \hat{f}(x)) \cdot \mathds{1}_{d(x, c(x)) \leq M}] \leq
  \Exp[\lossq(x, f(x))] + M \cdot K_Q
\]

To show the RHS we observe that
\begin{align*}
\frac{\sup_x |\bar{B}_M(x)|}{m} &\Exp[\ell_T(x, x_p, x_n)] \\
&\geq \frac{1}{m} \Exp_{x, x_p}\left[ \sum_{x_n' \in \bar{B}_M(x)} \frac{\sup_x |\bar{B}_M(x)|}{|\bar{B}_M(x_n')|} \ell_T(x, x_p, x_n') \right] \\
&\geq \frac{1}{m} \Exp\left[ \sup_{x_n' \in \bar{B}_M(x)} \ell_T(x, x_p, x_n') \right] \\
&\geq \Exp\left[ \inf_{x_n^* \in \bar{B}_M(x)} \frac{1}{m}\ell_T(x, x_p, x_n^*) \right] \\
&\geq \Exp\left[ \mathds{1}_{\inf_{x' \in \bar{B}_M(x)} |\emb(x) - \emb(x')| \leq |\emb(x) - \emb(x_p)|} \right] \\
&\geq \mathbb{P}[d(x, c(x)) > M]
\end{align*}
which follow from Lemmas \ref{lemma:prob-bound} and \ref{lemma:triplet-bound}.

Taking expectations, using H\"{o}lder's inequality, and maximizing $\lossq$
gives us the result.

\end{proof}

%
%
%
%
%

\bibliographystyle{ACM-Reference-Format}
\bibliography{paper}

\end{document}